\documentclass[11pt]{article}
\usepackage[top=1.4 in, bottom=1.4 in, left=1.15 in, right=1.15 in]{geometry}
\usepackage{amsmath,amsthm,amssymb}
\usepackage{algorithm}
\usepackage[noend]{algorithmic}
\usepackage{multirow}
\usepackage{url}
\usepackage[auth-sc]{authblk}
\usepackage{tikz}
\usetikzlibrary{arrows,shapes}

\newtheorem{theorem}{Theorem}
\newtheorem{corollary}{Corollary}

\title{\textbf{Speed-scaling with no Preemptions}}

\author[1]{Evripidis Bampis}
\author[2]{Dimitrios Letsios}
\author[1]{Giorgio Lucarelli}
\affil[1]{LIP6, Universit\'{e} Pierre et Marie Curie, France

\texttt{\{Evripidis.Bampis, Giorgio.Lucarelli\}@lip6.fr}\medskip }
\affil[2]{Institut f\"{u}r Informatik, Technische Universit\"{a}t M\"{u}nchen, Germany
\texttt{letsios@informatik.tu-muenchen.de}}

\begin{document}

\maketitle

\begin{abstract}
We revisit the non-preemptive speed-scaling problem,
in which a set of jobs have to be executed on a single or a set of parallel speed-scalable processor(s)
between their release dates and deadlines so that the energy consumption to be minimized.
We adopt the speed-scaling mechanism first introduced in [Yao et al., FOCS 1995]
according to which the power dissipated is a convex function of the processor's speed.
Intuitively, the higher is the speed of a processor, the higher is the energy consumption.
For the single-processor case, we improve the best known approximation algorithm by providing a $(1+\epsilon)^{\alpha}\tilde{B}_{\alpha}$-approximation algorithm,
where $\tilde{B}_{\alpha}$ is a generalization of the Bell number.
For the multiprocessor case, we present an approximation algorithm of ratio
$\tilde{B}_{\alpha}((1+\epsilon)(1+\frac{w_{\max}}{w_{\min}}))^{\alpha}$
improving the  best known result  by a factor of $(\frac{5}{2})^{\alpha-1}(\frac{w_{\max}}{w_{\min}})^{\alpha}$.
Notice that our result holds for the fully heterogeneous environment
while the previous known result holds only in the more restricted case of parallel processors with identical power functions.
\end{abstract}

\section{Introduction}

Speed-scaling (or dynamic voltage scaling) is one of the main mechanisms to save energy in modern computing systems.
According to this mechanism, the speed of each processor may dynamically change over time,
while the energy consumed by the processor is proportional to a convex function of the speed,
capturing in this way the intuition that higher speeds lead to higher energy consumption.
More precisely, if the speed of a processor is equal to $s(t)$ at a time instant $t$,
then the power dissipated is $P(s(t))=s(t)^{\alpha}$, where $\alpha>1$ is a small constant.
For example, the value of $\alpha$ is theoretically between two and three for CMOS devices,
while some experimental studies showed that $\alpha$ is rather smaller:
1.11 for Intel PXA 270, 1.62 for Pentium M770 and 1.66 for a TCP offload engine \cite{WAT09}.
The energy consumption is the integral of the power over time, i.e., $E=\int P(s(t)) dt$.

In order to handle the energy consumption in a computing system with respect to the speed-scaling mechanism,
we consider the following scheduling problem.
We are given a set of jobs and a single processor or a set of parallel processors.
Each job is characterized by a release date, a deadline and an amount of workload
that has to be executed between the job's release date and deadline.
The objective is to find a feasible schedule that minimizes the energy consumption.
In order to describe such a feasible schedule,
we have to determine not only the job that has to be executed on every machine at each time instant, but also the speed of each processor.

Speed-scaling scheduling problems have been extensively studied in the literature.
Since the seminal paper by Yao et al. \cite{YDS95} on 1995 until very recently,
all the energy minimization works considered the \emph{preemptive} case
in which the execution of a job may be interrupted and restarted later on the same or even on a different processor (\emph{migratory} case).
Only the last three years, some works study the \emph{non-preemptive} case.
In this paper, we improve the best known approximation algorithms for the non-preemptive case
for both the single-processor and the multiprocessor environments.

\paragraph{Problem definition and notation.}

We consider a set $\mathcal{J}$ of $n$ jobs, each one cha\-ra\-cterized by an amount of work $w_j$, a release date $r_j$ and a deadline $d_j$.
We will consider both the single-processor and the multiprocessor cases.
In the multiprocessor environment, we denote by $\mathcal{P}$ the set of $m$ available parallel processors.
Moreover, we distinguish between the \emph{homogeneous} and the \emph{heterogeneous} multiprocessor cases.
In the latter one, we assume that each processor $i \in \mathcal{P}$ has a different constant $\alpha_i$,
capturing in this way the existence of processors with different energy consumption rate.
For simplicity, we define $\alpha=\max_{i \in \mathcal{P}}\{\alpha_i\}$.
Moreover, in the \emph{fully heterogeneous} case we additionally assume that each job $j \in \mathcal{J}$
has a different work $w_{i,j}$, release date $r_{i,j}$ and deadline $d_{i,j}$ on each processor $i \in \mathcal{P}$.
In all cases, the objective is to find a schedule that minimizes the energy consumption with respect to the speed-scaling mechanism,
such that each job $j \in \mathcal{J}$ is executed during its \emph{life} interval $[r_j,d_j]$.
The results presented in this paper assume that the preemption of jobs is not allowed;
and hence neither their migration in the multiprocessor environments.

In what follows, we denote by $w_{\max}$ and $w_{\min}$ the maximum and the minimum work, respectively, among all jobs.
Moreover, we call an instance \emph{agreeable} if earlier released jobs have earlier deadlines,
i.e., for each $j$ and $j'$ with $r_j \leq r_{j'}$ then $d_j \leq d_{j'}$.
Finally, given a schedule $\mathcal{S}$ we denote by $E(\mathcal{S})$ its energy consumption.

\paragraph{Related work.}

In \cite{YDS95}, a polynomial-time algorithm has been presented that finds an optimal preemptive schedule when a single processor is available.
In the case where the preemption and also the migration of jobs are allowed,
several polynomial-time algorithms have been proposed when a set of homogeneous parallel processors is available \cite{AAG11,ABKL12,BLL12,BG08},
while in the fully heterogeneous environment an $OPT+\epsilon$ algorithm with complexity polynomial to $\frac{1}{\epsilon}$ has been presented in \cite{BKLLS13}.
In the case where the preemption of jobs is allowed but not their migration,
the problem becomes strongly NP-hard even if all jobs have equal release dates and equal deadlines \cite{AMS07}.
For this special case, the authors in \cite{AMS07} observed that a PTAS can be derived from \cite{HS87}.
For arbitrary release dates and deadlines, a $B_{\lceil\alpha\rceil}$-approximation algorithm is known \cite{GNS09},
where $B_{\lceil\alpha\rceil}$ is the $\lceil\alpha\rceil$-th Bell number.
This result has been extended in \cite{BKLLS13} for the fully heterogeneous environment,
where an approximation algorithm of ratio $(1+\epsilon)^{\alpha}\tilde{B}_{\alpha}$ has been presented, where
$\tilde{B}_{\alpha}=\sum_{k=0}^{\infty} \frac{k^{\alpha}e^{-1}}{k!}$ is a generalization of the Bell number that is also valid  for fractional values of $\alpha$.

When preemptions are not allowed, Antoniadis and Huang \cite{AH12} proved that the single-processor case is strongly NP-hard,
while they have also presented a $2^{5\alpha-4}$-approximation algorithm.
In \cite{BKLLS13}, an approximation algorithm of ratio $2^{\alpha-1}(1+\epsilon)^{\alpha}\tilde{B}_{\alpha}$ has been proposed,
improving the ratio given in \cite{AH12} for any $\alpha<114$.
Recently, an approximation algorithm of ratio $(12(1+\epsilon))^{\alpha-1}$ is given in \cite{CLMM14},
improving the approximation ratio for any $\alpha>25$.
Moreover, the relation between preemptive and non-preemptive schedules in the energy-minimization setting has been studied in \cite{BKLLN13}. The authors show that using as a lower bound the optimal preemptive solution, it is possible to obtain
an approximation algorithm of ratio $(1+\frac{w_{\max}}{w_{\min}})^{\alpha}$.
In the special case where all jobs have equal work this leads to a constant factor approximation of $2^{\alpha}$.
Recently, for this special case, Angel et al. \cite{ABC13} and Huang and Ott \cite{HO14}, independently, proposed
an optimal polynomial-time algorithm based on dynamic programming.
Note also that for agreeable instances the single-processor non-preemptive speed-scaling problem can be solved to optimality in polynomial time,
as the algorithm proposed by Yao et al. \cite{YDS95} for the preemptive case
returns a non-preemptive schedule for agreeable instances.

For homogeneous multiprocessors when preemptions are not allowed,
an approximation algorithm with ratio $m^{\alpha}(\sqrt[m]{n})^{\alpha-1}$ has been presented in \cite{BKLLN13}
which uses as a lower bound the optimal preemptive solution.
More recently, Cohen-Addad et al. \cite{CLMM14} proposed an algorithm of ratio
$(\frac{5}{2})^{\alpha-1}\tilde{B}_{\alpha}((1+\epsilon)(1+\frac{w_{\max}}{w_{\min}}))^{\alpha}$,
transforming the problem to the fully heterogeneous preemptive non-migratory case and using the approximation algorithm proposed in \cite{BKLLS13}.
This algorithm leads to an approximation ratio of $2(1+\epsilon)^{\alpha}5^{\alpha-1}\tilde{B}_{\alpha}$ for the case where all jobs have equal work.
The authors in \cite{CLMM14} observe also that their algorithm can be used
when each job $j \in \mathcal{J}$ has a different work $w_{i,j}$ on each processor $i \in \mathcal{P}$,
by loosing an additional factor of $(\frac{w_{\max}}{w_{\min}})^{\alpha}$.

Several other results concerning scheduling problems in the speed-scaling setting have been presented,
involving the optimization of some Quality of Service (QoS) criterion under a budget of energy,
or the optimization of a linear combination of the energy consumption and some QoS criterion (see for example \cite{BLMZ12,B06,PSU08}).
Moreover, two other energy minimization variants of the speed-scaling model have been studied in the literature,
namely the \emph{bounded speed model} in which the speeds of the processors are bounded above and below (see for example \cite{CCLLMW07}),
and the \emph{discrete speed model} in which the speeds of the processors can be selected
among a set of discrete speeds (see for example \cite{LY06}).
The interested reader can find more details in the surveys \cite{A10,A11}.

\paragraph{Our contribution.}

In Section~\ref{section:single} we revisit the single-processor non-preemptive speed-scaling problem,
and we present an approximation algorithm of ratio $(1+\varepsilon)^{\alpha-1}\tilde{B}_{\alpha}$ which becomes the best algorithm for any $\alpha \leq 77$.
Recall that in practice $\alpha$ is a small constant and usually $\alpha \in (1,3]$.
In \cite{BKLLN13}, where  the relation between preemptive and non-preemptive schedules has been explored,
an example showing that their distance can be a factor of $\Omega(n^{\alpha-1})$ have been proposed.
A similar example was used in \cite{CLMM14} to show that the standard configuration linear programming formulation has the same integrality gap.
In both cases, $w_{\max}=n$ and $w_{\min}=1$ and
the distance between the optimal preemptive and non-preemptive schedules can be seen as $\Omega((\frac{w_{\max}}{w_{\min}})^{\alpha-1})$.
In this direction, a $(1+\frac{w_{\max}}{w_{min}})^{\alpha-1}$-approximation algorithm for the single-processor case has been presented in \cite{BKLLN13}.
To overcome the above lower bound, all known constant-factor approximation algorithms for the single-processor problem \cite{AH12,BKLLS13,CLMM14}
consider an initial partition of the time horizon into some specific intervals defined by the so-called \emph{landmarks}.
These intervals are defined in such a way that there is not a job whose life interval is completely included in one of them.
Intuitively, this partition is used in order to improve the lower bound by focusing on special preemptive schedules
that can be transformed to a feasible non-preemptive schedules without loosing a lot in terms of approximation.
Here, we are able to avoid the use of this partition improving in this way the result of \cite{BKLLS13} by a factor of $2^{\alpha-1}$.
In order to do that, we modify the \emph{configuration linear program} proposed in \cite{BKLLS13}
by including an additional structural property that is valid for any feasible non-preemptive schedule.
This property helps us to obtain a ``good'' preemptive schedule after a randomized rounding procedure. We transform this ``good'' preemptive schedule to a new instance of the energy-minimization single-processor problem that is  \emph{agreeable} by choosing in an appropriate way new release dates and deadlines for the jobs.
In this way, it is then sufficient to apply the Earliest Deadline First policy in order to get a non-preemptive schedule of the same energy consumption as the preemptive one.

In Section~\ref{section:parallel} we consider the fully heterogeneous non-preemptive speed-scaling problem, and we improve the approximation ratio of
$(\frac{w_{\max}}{w_{\min}})^{\alpha}(\frac{5}{2})^{\alpha-1}\tilde{B}_{\alpha}((1+\epsilon)(1+\frac{w_{\max}}{w_{\min}}))^{\alpha}$ given in \cite{CLMM14}
to $\tilde{B}_{\alpha}((1+\epsilon)(1+\frac{w_{\max}}{w_{\min}}))^{\alpha}$.
Consecutively, our result generalizes and improves the approximation ratio for the equal-works case from $2(1+\epsilon)^{\alpha}5^{\alpha-1}\tilde{B}_{\alpha}$
to $(2(1+\epsilon))^{\alpha}\tilde{B}_{\alpha}$.
Note also that we generalize the machine environment and we pass from the homogeneous with different $w_{i,j}$'s to the fully heterogeneous one.
Our algorithm combines two basic ingredients: the $\tilde{B}_{\alpha}(1+\epsilon)^{\alpha}$-approximation algorithm of \cite{BKLLS13}
for the fully heterogeneous preemptive non-migratory speed-scaling problem
and the $(1+\frac{w_{\max}}{w_{\min}})^{\alpha}$-approximation algorithm of \cite{BKLLN13} for the single-processor non-preemptive speed-scaling problem.
The first algorithm is used in order to assign the jobs to the processors,
while the second one to get a non-preemptive schedule for each processor independently.
The key observation here is that the algorithm for the single-processor non-preemptive case presented in \cite{BKLLN13}
uses as a lower bound the optimal preemptive schedule.

We summarize our results with respect to the existing bibliography in Table~\ref{tbl:results}.

\begin{table}
\begin{center}
\begin{tabular}{l||l|l}
Machine environment & Previous known result & Our results\\
\hline
\hline
\multirow{2}{*}{single-processor} & $2^{\alpha-1}(1+\epsilon)^{\alpha}\tilde{B}_{\alpha}$ \cite{BKLLS13} & \multirow{2}{*}{$(1+\epsilon)^{\alpha}\tilde{B}_{\alpha}$}\\
 & $(12(1+\epsilon))^{\alpha-1}$ \cite{CLMM14} & \\
\hline
homogeneous & $(\frac{5}{2})^{\alpha-1}\tilde{B}_{\alpha}((1+\epsilon)(1+\frac{w_{\max}}{w_{\min}}))^{\alpha}$ \cite{CLMM14} & \\
homogeneous with $w_{i,j}$'s & $(\frac{5}{2})^{\alpha-1}\tilde{B}_{\alpha}((1+\epsilon)(1+\frac{w_{\max}}{w_{\min}})\frac{w_{\max}}{w_{\min}})^{\alpha}$ \cite{CLMM14}
    & \\
fully heterogeneous & & $\tilde{B}_{\alpha}((1+\epsilon)(1+\frac{w_{\max}}{w_{\min}}))^{\alpha}$
\end{tabular}
\end{center}
\label{tbl:results}
\caption{Comparison of the approximation ratios obtained in this paper with the previously best known approximation ratios.}
\end{table}

\section{Single-processor}
\label{section:single}

In this section we consider the single-processor non-preemptive case
and we present an approximation algorithm of ratio $(1+\varepsilon)\tilde{B}_{\alpha}$,
improving upon the previous known results \cite{AH12,BKLLS13,CLMM14} for any $\alpha \leq 77$.
Our algorithm is based on a linear programming formulation combining ideas from \cite{BKLLS13,CLMM14}
and the randomized rounding proposed in \cite{BKLLS13}.

Before formulating the problem as a linear program we need to discretize the time into slots.
Consider the set of all different release dates and deadlines of jobs in increasing order, i.e., $t_1 < t_2 < \ldots < t_k$.
For each $\ell$ , $1 \leq \ell \leq k-1$, we split the time between $t_{\ell}$ and $t_{\ell+1}$
into $n^2(1+\frac{1}{\epsilon})$ equal length slots as proposed in \cite{HO14}.
Let $\mathcal{T}$ be the set of all created slots.
Henceforth, we will consider only solutions in which each slot can be occupied by at most one job which uses the whole slot.
Huang and Ott \cite{HO14} proved that this can be done by loosing a factor of $(1+\epsilon)^{\alpha-1}$.

Our formulation is based on the configuration linear program which was proposed in \cite{BKLLS13}.
In \cite{CLMM14}, an additional constraint was used for the single-processor non-preemptive problem.
This constraint implies that the life interval of a job cannot be included to the execution interval of another job.
We explicitly incorporate this constraint in the definition of the set of configurations for each job.
More specifically, for a job $j \in \mathcal{J}$, we define a configuration $c$ to be a set of consecutive slots in $[r_j,d_j]$
such that there is not another job $j'$ whose life interval $[r_{j'},d_{j'}]$ is completely included in $c$.
Let $\mathcal{C}_j$ be the set of all possible configurations for the job $j$.
We introduce a binary variable $x_{j,c}$ which is equal to one if the job $j$ is executed according the configuration $c \in \mathcal{C}_j$.
Let $|c|$ be the length (in time) of the configuration $c$.
Note that the number of configurations is polynomial as they only contain consecutive slots and the number of slots is also polynomial.

For notational convenience, we write $t \in c$ if the slot $t \in \mathcal{T}$ is part of the configuration $c \in \mathcal{C}_j$ of job $j \in \mathcal{J}$.
By the convexity of the power function, each job in an optimal schedule runs in a constant speed (see for example \cite{YDS95}).
Hence, the quantity $\frac{w_j^{\alpha}}{|c|^{\alpha-1}}$ corresponds to the energy consumed by $j$ if it is executed according to $c$,
as the constant speed that will be used for $j$ is equal to $\frac{w_j}{|c|}$.
Consider the following integer linear program.

\begin{eqnarray*}
\min \sum_{j \in \mathcal{J}} \sum_{c \in \mathcal{C}_j} x_{j,c} \frac{w_j^{\alpha}}{|c|^{\alpha-1}}\\
\sum_{c \in \mathcal{C}_j} x_{j,c} \geq 1 & & \forall j \in \mathcal{J}\\
\sum_{j \in \mathcal{J}} \sum_{c \in \mathcal{C}_j: t \in c} x_{j,c} \leq 1 & & \forall t \in \mathcal{T}\\
x_{j,c} \in \{0,1\} & & \forall j \in \mathcal{J}, c \in \mathcal{C}_j
\end{eqnarray*}

The first constraint ensures that each job is executed according to a configuration.
The second constraint implies that at each slot at most one configuration and hence at most one job can be executed.

We consider the randomized rounding procedure proposed in \cite{BKLLS13} for the fully heterogeneous preemptive non-migratory speed-scaling problem,
adapted to the single processor environment.
More specifically, for each job $j \in \mathcal{J}$ we choose at random with probability $x_{j,c}$ a configuration $c \in \mathcal{C}_{j}$.
By doing this, more than one jobs may be assigned in a slot $t \in \mathcal{T}$ which has as a result to get a non-feasible schedule.
In order to deal with this infeasibility, for each slot $t \in \mathcal{T}$ we perform an appropriate speed-up
that leads to a feasible preemptive schedule.
The above procedure is described formally in Algorithm~\ref{algo:single}.

\begin{algorithm}
\begin{algorithmic}[1]
\STATE Solve the configuration LP relaxation;
\STATE For each job $j \in \mathcal{J}$, choose a configuration at random with probability $x_{j,c}$;
\STATE Let $w_j(t)$ be the amount of work executed for job $j$ during the slot $t\in\mathcal{T}$ according to its chosen configuration;
\STATE Set the processor's speed during $t$ as if $\sum_{j\in\mathcal{J}}w_j(t)$ units of work are executed with constant speed during the entire $t$,
       i.e., $\sum_{j\in\mathcal{J}}w_j(t) / |t|$, where $|t|$ is the length of $t$;
\end{algorithmic}
\caption{}
\label{algo:single}
\end{algorithm}

The analysis of the above procedure in \cite{BKLLS13} is done independently for each slot,
while the speed-up performed leads to a loss of a factor of $\tilde{B}_{\alpha}$ to the approximation ratio.
We can use exactly the same analysis and get the same approximation guarantee for our problem.

In what follows, given the feasible preemptive schedule $\mathcal{S}_{pr}$ obtained by Algorithm~\ref{algo:single},
we will create a feasible non-preemptive schedule $\mathcal{S}_{npr}$ of the same energy consumption.
In fact, we first create a restricted agreeable instance of our initial instance based on $\mathcal{S}_{pr}$.
Then, we will apply the Earliest Deadline First (EDF) policy to get a non-preemptive schedule.

Consider a job $j \in \mathcal{J}$.
Let $p_j$ be the total processing time of $j$ in $\mathcal{S}_{pr}$.
Moreover, assume that the first piece of $j$ in $\mathcal{S}_{pr}$ begins at $b_j$ and the last piece of $j$ in $\mathcal{S}_{pr}$ ends at $e_j$.
Note that the interval $[b_j,e_j] \subseteq [r_j,d_j]$ does not include the life interval of any other job, by the definition of configurations.
For the job $j$, we create an interval $[r'_j,d'_j]$, where $[b_j,e_j] \subseteq [r'_j,d'_j] \subseteq [r_j,d_j]$, as follows.
We select $r'_j$ such that $[r'_j,e_j]$ does not include the whole life interval of another job, $r_j \leq r'_j \leq b_j$ and $r'_j$ is minimum.
We then define $d'_j$ in a similar way, i.e.,
we select $d'_j$ such that $[r'_j,d'_j]$ does not include the whole life interval of another job, $e_j \leq d'_j \leq d_j$ and $d'_j$ is maximum.

An example of the above transformation is given in Fig.~\ref{fig:transformation}.
In this picture, the life intervals of jobs $J_1$ and $J_4$ are shortened.
For example, in the preemptive schedule $\mathcal{S}_{pr}$ the job $J_4$ is executed on the right of the job $J_5$.
Hence, in the restricted instance we cut down the part of the life interval of $J_4$ which is on the left of the release date of $J_5$.
Intuitively, we decide if $J_4$ should be executed on the left or on the right of $J_5$ with respect to $\mathcal{S}_{pr}$
and we transform the initial instance appropriately.

In order to see that the created instance is agreeable,
assume for contradiction that there are two jobs $i,j \in \mathcal{J}$ such that $r'_j < r'_i$ and $d'_i < d'_j$.
The algorithm does not further decrease $r'_i$ because either $[r'_i,d'_i]$ includes the life interval of a job $k \in \mathcal{J}$ or $r'_i=r_i$.
In the first case, we have that $[r_k,d_k] \subseteq [r'_j,d'_j]$,
which is a contradiction to the definition of configurations and the selection of $r'_j$.
In the second case, the algorithm does not further increase $d'_i$
because either $[r'_i,d'_i]$ includes the life interval of a job $\ell \in \mathcal{J}$ or $d'_i=d_i$.
In both last cases we have again a contradiction, as either $[r_{\ell},d_{\ell}] \subseteq [r'_j,d'_j]$ or $[r_i,d_i] \subseteq [r'_j,d'_j]$.

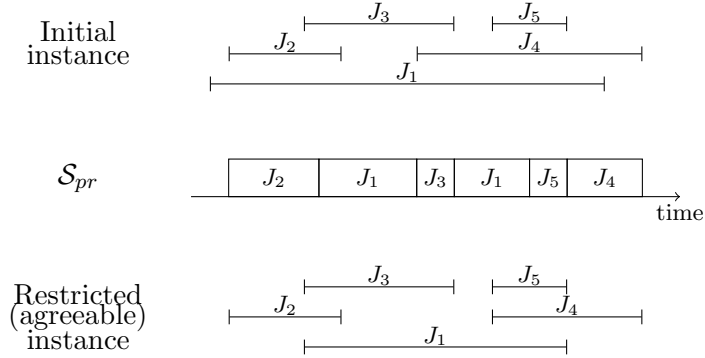
\begin{figure}[htb]
\begin{center}
\begin{tikzpicture}
\node at (-1.5,0.7) {Initial};
\node at (-1.5,0.4) {instance};
\draw[|-|] (0.25,0.0) -- (5.5,0.0);
\node[font=\footnotesize] at (2.875,0.15) {$J_1$};
\draw[|-|] (0.5,0.4) -- (2.0,0.4);
\node[font=\footnotesize] at (1.25,0.55) {$J_2$};
\draw[|-|] (3.0,0.4) -- (6.0,0.4);
\node[font=\footnotesize] at (4.5,0.55) {$J_4$};
\draw[|-|] (1.5,0.8) -- (3.5,0.8);
\node[font=\footnotesize] at (2.5,0.95) {$J_3$};
\draw[|-|] (4.0,0.8) -- (5.0,0.8);
\node[font=\footnotesize] at (4.5,0.95) {$J_5$};
\node at (-1.5,-1.25) {$\mathcal{S}_{pr}$};
\draw[->] (0.0,-1.5) -- (6.5,-1.5);
\node[font=\footnotesize] at (6.5,-1.7) {time};
\draw (0.5,-1.5) rectangle (1.7,-1.0);
\node[font=\footnotesize] at (1.1,-1.25) {$J_2$};
\draw (1.7,-1.5) rectangle (3.0,-1.0);
\node[font=\footnotesize] at (2.35,-1.25) {$J_1$};
\draw (3.0,-1.5) rectangle (3.5,-1.0);
\node[font=\footnotesize] at (3.25,-1.25) {$J_3$};
\draw (3.5,-1.5) rectangle (4.5,-1.0);
\node[font=\footnotesize] at (4.0,-1.25) {$J_1$};
\draw (4.5,-1.5) rectangle (5.0,-1.0);
\node[font=\footnotesize] at (4.75,-1.25) {$J_5$};
\draw (5.0,-1.5) rectangle (6.0,-1.0);
\node[font=\footnotesize] at (5.5,-1.25) {$J_4$};
\node at (-1.5,-2.8) {Restricted};
\node at (-1.5,-3.1) {(agreeable)};
\node at (-1.5,-3.4) {instance};
\draw[|-|] (1.5,-3.5) -- (5.0,-3.5);
\node[font=\footnotesize] at (3.25,-3.35) {$J_1$};
\draw[|-|] (0.5,-3.1) -- (2.0,-3.1);
\node[font=\footnotesize] at (1.25,-2.95) {$J_2$};
\draw[|-|] (4.0,-3.1) -- (6.0,-3.1);
\node[font=\footnotesize] at (5.0,-2.95) {$J_4$};
\draw[|-|] (1.5,-2.7) -- (3.5,-2.7);
\node[font=\footnotesize] at (2.5,-2.55) {$J_3$};
\draw[|-|] (4.0,-2.7) -- (5.0,-2.7);
\node[font=\footnotesize] at (4.5,-2.55) {$J_5$};
\end{tikzpicture}
\caption{The transformation of an instance into a restricted (agreeable) instance based on the feasible preemptive schedule $\mathcal{S}_{pr}$.}
\label{fig:transformation}
\end{center}
\end{figure}

So, we create a restricted agreeable instance using $r'_j$ and $d'_j$ as the release date and the deadline, respectively, of the job $j \in \mathcal{J}$.
Moreover, we set the processing time of $j$ to $p_j$.
The life interval of each job $j \in \mathcal{J}$ in the created instance
is a superset of its execution interval to $\mathcal{S}$, i.e., $[b_j,e_j] \subseteq [r'_j,d'_j]$.
Hence, the schedule $\mathcal{S}$ is a feasible preemptive schedule for the restricted instance.
Moreover, the life interval of each job $j \in \mathcal{J}$ in the restricted instance
is a subset of its life interval to the initial instance, i.e., $[r'_j,d'_j] \subseteq [r_j,d_j]$.
Thus, we can get a feasible non-preemptive schedule $\mathcal{S}_{npr}$ for our initial instance, by applying the EDF policy to the new instance.
Finally, as in both schedules $\mathcal{S}_{pr}$ and $\mathcal{S}_{npr}$ we use the some processing times, their energy consumption is the same.

\begin{theorem}
There is a $(1+\varepsilon)^{\alpha-1}\tilde{B}_{\alpha}$-approximation algorithm for the single-processor non-preemptive speed-scaling problem.
\end{theorem}

\section{Parallel Processors}
\label{section:parallel}

In this section we consider the fully heterogeneous multiprocessor case and we propose an approximation algorithm
of ratio $\tilde{B}_{\alpha}((1+\epsilon)(1+\frac{w_{\max}}{w_{\min}}))^{\alpha}$,
generalizing the recent result by Cohen-Addad et al. \cite{CLMM14} from the homogeneous with different $w_{i,j}$'s to the fully heterogeneous environment
and improving their ratio by a factor of $(\frac{w_{\max}}{w_{\min}})^{\alpha}(\frac{5}{2})^{\alpha-1}$.
Our algorithm uses the following result proposed in \cite{BKLLN13}.

\begin{theorem}\label{thm:cocoon}
\emph{\cite{BKLLN13}} There is an approximation algorithm for the single-processor non-preemptive speed-scaling problem
that returns a schedule $\mathcal{S}$ with energy consumption
$$E(\mathcal{S}) \leq (1+\frac{w_{\max}}{w_{\min}})^{\alpha} E(\mathcal{S}^*_{pr}) \leq (1+\frac{w_{\max}}{w_{\min}})^{\alpha} E(\mathcal{S}^*_{npr})$$
where $\mathcal{S}^*_{pr}$ and $\mathcal{S}^*_{npr}$ are the optimal schedules for the preemptive and the non-preemptive case, respectively.
\end{theorem}

The key observation in the above theorem concerns the intermediate result that
the approximation ratio of the algorithm for the single-processor non-preemptive case proposed in \cite{BKLLN13}
uses as lower bound the optimal preemptive schedule.
Based on this, we propose Algorithm~\ref{algo:heterogeneous} which uses the
$(1+\epsilon)^{\alpha}\tilde{B}_{\alpha}$-approximation algorithm proposed in \cite{BKLLS13}
for the fully heterogeneous preemptive non-migratory speed-scaling problem to find a good assignment of the jobs to the processors
and then applies Theorem~\ref{thm:cocoon} to create a non-preemptive schedule independently for each processor.

\begin{algorithm}
\begin{algorithmic}[1]
\STATE Find a preemptive non-migratory schedule $\mathcal{S}$ using the algorithm proposed in \cite{BKLLS13} for the fully heterogeneous environment;
\FOR {each processor $i \in \mathcal{P}$}
\STATE Let $\mathcal{J}_i$ be the set of jobs assigned to processor $i$ according to $\mathcal{S}$;
\STATE Find a single-processor non-preemptive schedule $\mathcal{S}_{i,npr}$
       using the algorithm proposed in \cite{BKLLN13} (Theorem~\ref{thm:cocoon}) with input $\mathcal{J}_i$;
\ENDFOR
\RETURN the non-preemptive schedule $\mathcal{S}_{npr}$ which consists of the non-preemptive schedules $\mathcal{S}_{i,npr}$, $1 \leq i \leq m$;
\end{algorithmic}
\caption{}
\label{algo:heterogeneous}
\end{algorithm}

\begin{theorem}\label{thm:heterogeneous}
Algorithm~\ref{algo:heterogeneous} achieves an approximation ratio of $\tilde{B}_{\alpha}((1+\epsilon)(1+\frac{w_{\max}}{w_{\min}}))^{\alpha}$
for the fully heterogeneous non-preemptive speed-scaling problem.
\end{theorem}
\begin{proof}
Consider first the schedule $\mathcal{S}$ obtained in Line~1 of the algorithm,
and let $\mathcal{S}_{i,pr}$ be the (sub)schedule of $\mathcal{S}$ that corresponds to the processor $i \in \mathcal{P}$.
In other words, each $\mathcal{S}_{i,pr}$ is a feasible preemptive schedule of the subset of jobs $\mathcal{J}_i$.
As $\mathcal{S}$ is a non-migratory schedule the subsets of jobs $\mathcal{J}_1,\mathcal{J}_2,\ldots,\mathcal{J}_m$ are pairwise disjoint.
Hence, we have that
\begin{equation}\label{eq:h1}
\sum_{i \in \mathcal{P}} E(\mathcal{S}_{i,pr}) = E(\mathcal{S}) \leq (1+\epsilon)^{\alpha}\tilde{B}_{\alpha} E(\mathcal{S}^*)
\end{equation}
where $\mathcal{S}^*$ is the optimal non-preemptive schedule for our problem and
the inequality holds by the result in \cite{BKLLS13} and the fact that the energy consumption in an optimal preemptive-non-migratory schedule
is a lower bound to the energy consumption of $\mathcal{S}^*$.

Consider now, for each processor $i \in \mathcal{P}$, the schedule $\mathcal{S}_{i,npr}$ created in Line~4 of the algorithm.
By Theorem~\ref{thm:cocoon} we have that
\begin{equation*}
E(\mathcal{S}_{i,npr}) \leq \left(1+\frac{w_{\max}}{w_{\min}}\right)^{\alpha} E(\mathcal{S}^*_{i,pr})
\end{equation*}
where $\mathcal{S}^*_{i,pr}$ is an optimal preemptive schedule for the subset of jobs $\mathcal{J}_i$.
As $\mathcal{S}^*_{i,pr}$ and $\mathcal{S}_{i,pr}$ are schedules concerning the same set of jobs
and $\mathcal{S}^*_{i,pr}$ is the optimal preemptive schedule, we have that
\begin{equation}\label{eq:h2}
E(\mathcal{S}_{i,npr}) \leq \left(1+\frac{w_{\max}}{w_{\min}}\right)^{\alpha} E(\mathcal{S}_{i,pr})
\end{equation}

Since $\mathcal{S}_{npr}$ is the concatenation of $\mathcal{S}_{i,npr}$ for all $i \in \mathcal{P}$,
and by using Equations~(\ref{eq:h1}) and~(\ref{eq:h2}), the theorem follows.
\end{proof}

Algorithm~\ref{algo:heterogeneous} can be also used for the case where all jobs have equal work on each processor,
i.e., each job $j \in \mathcal{J}$ has to execute an amount of work $w_{i,j}=w_i$ if it is assigned on processor $i \in \mathcal{P}$.
In this case we get the following result.

\begin{corollary}
Algorithm~\ref{algo:heterogeneous} achieves a constant-approximation ratio of $\tilde{B}_{\alpha}(2(1+\epsilon))^{\alpha}$
for the fully heterogeneous non-preemptive speed-scaling problem when all jobs have equal work on each processor.
\end{corollary}

\section{Conclusions}
In this paper, we have presented algorithms with improved approximation ratios for both the single-processor and the multiprocessor environments.
A challenging question left open in this work is the existence of a constant approximation ratio algorithm for the multiprocessor case.
Also, there is a need for non-approximability results in the same vein as the one presented in \cite{CLMM14}.

%\bibliographystyle{plain}
%\bibliography{bibfile}

\begin{thebibliography}{10}

\bibitem{A10}
S.~Albers.
\newblock Energy-efficient algorithms.
\newblock {\em Communications of the ACM}, 53(5):86--96, 2010.

\bibitem{A11}
S.~Albers.
\newblock Algorithms for dynamic speed scaling.
\newblock In {\em STACS}, volume~9 of {\em LIPIcs}, pages 1--11. Schloss
  Dagstuhl - Leibniz-Zentrum fuer Informatik, 2011.

\bibitem{AAG11}
S.~Albers, A.~Antoniadis, and G.~Greiner.
\newblock On multi-processor speed scaling with migration: extended abstract.
\newblock In {\em SPAA}, pages 279--288. ACM, 2011.

\bibitem{AMS07}
S.~Albers, F.~M\"{u}ller, and S.~Schmelzer.
\newblock Speed scaling on parallel processors.
\newblock In {\em SPAA}, pages 289--298. ACM, 2007.

\bibitem{ABC13}
E.~Angel, E.~Bampis, and V.~Chau.
\newblock Throughput maximization in the speed-scaling setting.
\newblock {\em CoRR}, abs/1309.1732, 2013.

\bibitem{ABKL12}
E.~Angel, E.~Bampis, F.~Kacem, and D.~Letsios.
\newblock Speed scaling on parallel processors with migration.
\newblock In {\em Euro-Par}, volume 7484 of {\em LNCS}, pages 128--140, 2012.

\bibitem{AH12}
A.~Antoniadis and C.-C. Huang.
\newblock Non-preemptive speed scaling.
\newblock In {\em SWAT}, volume 7357 of {\em LNCS}, pages 249--260. Springer,
  2012.

\bibitem{BKLLN13}
E.~Bampis, A.~Kononov, D.~Letsios, G.~Lucarelli, and I.~Nemparis.
\newblock From preemptive to non-preemptive speed-scaling scheduling.
\newblock In {\em COCOON}, volume 7936 of {\em LNCS}, pages 134--146. Springer,
  2013.

\bibitem{BKLLS13}
E.~Bampis, A.~Kononov, D.~Letsios, G.~Lucarelli, and M.~Sviridenko.
\newblock Energy efficient scheduling and routing via randomized rounding.
\newblock In {\em FSTTCS}, volume~24 of {\em LIPIcs}, pages 449--460. Schloss
  Dagstuhl - Leibniz-Zentrum fuer Informatik, 2013.

\bibitem{BLL12}
E.~Bampis, D.~Letsios, and G.~Lucarelli.
\newblock Green scheduling, flows and matchings.
\newblock In {\em ISAAC}, volume 7676 of {\em LNCS}, pages 106--115. Springer,
  2012.

\bibitem{BLMZ12}
E.~Bampis, D.~Letsios, I.~Milis, and G.~Zois.
\newblock Speed scaling for maximum lateness.
\newblock In {\em COCOON}, volume 7434 of {\em LNCS}, pages 25--36. Springer,
  2012.

\bibitem{BG08}
B.~D. Bingham and M.~R. Greenstreet.
\newblock Energy optimal scheduling on multiprocessors with migration.
\newblock In {\em ISPA}, pages 153--161. IEEE, 2008.

\bibitem{B06}
D.~P. Bunde.
\newblock Power-aware scheduling for makespan and flow.
\newblock In {\em SPAA}, pages 190--196. ACM, 2006.

\bibitem{CCLLMW07}
H.-L. Chan, W.-T. Chan, T.~W. Lam, L.-K. Lee, K.-S. Mak, and P.~W.~H. Wong.
\newblock Energy efficient online deadline scheduling.
\newblock In {\em SODA}, pages 795--804, 2007.

\bibitem{CLMM14}
V.~Cohen-Addad, Z.~Li, C.~Mathieu, and I.~Milis.
\newblock Energy-efficient algorithms for non-preemptive speed-scaling.
\newblock {\em CoRR}, abs/1402.4111, 2014.

\bibitem{GNS09}
G.~Greiner, T.~Nonner, and A.~Souza.
\newblock The bell is ringing in speed-scaled multiprocessor scheduling.
\newblock In {\em SPAA}, pages 11--18. ACM, 2009.

\bibitem{HS87}
D.S. Hochbaum and D.B. Shmoys.
\newblock Using dual approximation algorithms for scheduling problems:
  {T}heoretical and practical results.
\newblock {\em Journal of the ACM}, 34:144--162, 1987.

\bibitem{HO14}
C.-C. Huang and S.~Ott.
\newblock New results for non-preemptive speed scaling.
\newblock In {\em MFCS}, LNCS. Springer, 2014.

\bibitem{LY06}
M.~Li and F.~F. Yao.
\newblock An efficient algorithm for computing optimal discrete voltage
  schedules.
\newblock {\em SIAM Journal on Computing}, 35:658--671, 2006.

\bibitem{PSU08}
K.~Pruhs, R.~van Stee, and P.~Uthaisombut.
\newblock Speed scaling of tasks with precedence constraints.
\newblock {\em Theory of Computing Systems}, 43:67--80, 2008.

\bibitem{WAT09}
A.~Wierman, L.~L.~H. Andrew, and A.~Tang.
\newblock Power-aware speed scaling in processor sharing systems.
\newblock In {\em INFOCOM}, pages 2007--2015. IEEE, 2009.

\bibitem{YDS95}
F.~F. Yao, A.~J. Demers, and S.~Shenker.
\newblock A scheduling model for reduced {CPU} energy.
\newblock In {\em FOCS}, pages 374--382. IEEE Computer Society, 1995.

\end{thebibliography}

\end{document}